\documentclass[12pt]{extarticle}
\usepackage{amssymb, amsmath, amsfonts, amsthm, graphics}
\usepackage{inputenc}
\usepackage{lscape}
\usepackage{caption}
\usepackage{listings}
\usepackage{extsizes}

\DeclareMathOperator*{\argmin}{arg\,min}

{\bfseries}{\itshape}
\newtheorem*{Thm*}{Theorem}{\bfseries}{\itshape}
\newtheorem{Cor}{Corollary}{\bfseries}{\itshape}
\newtheorem{Prop}[Cor]{Proposition}{\bfseries}{\itshape}
{\bfseries}{\itshape}
\newtheorem*{Lem*}{Lemma}{\bfseries}{\itshape}
{\bfseries}{\itshape}
{\bfseries}{\itshape}
\newtheorem{Def}[Cor]{Definition}{\bfseries}{\rmfamily}
{\scshape}{\rmfamily}
\newtheorem{Rem}[Cor]{Remark}{\scshape}{\rmfamily}
{\bfseries}{\itshape}

\newcommand{\R}{\mathbb{R}}
\title{Calculation of sample size guaranteeing the required width of the empirical confidence interval with predefined probability.}
\author{Ilya Novikov. \\
  Biostatistical and Biomathematical unit \\
  Gertner Institute for Epidemiology \\
  and Health Policy Research \\
  Ramat Gan, Israel \\
  ilian@gertner.health.gov.il }

\begin{document}

\maketitle
\newpage
\begin{abstract}
The goal of any estimation study is an interval estimation of a the parameter(s) of interest. These estimations are mostly expressed using empirical confidence intervals that are based on sample point estimates of the corresponding parameter(s). In contrast, calculations of the necessary sample size usually use expected confidence intervals that are based on the expected value of the parameter(s). The approach that guarantees the required probability of the required width of empirical confidence interval is known at least since 1989. However, till now, this approach is not implemented for most software and is not even described in many modern papers and textbooks. Here we present the concise description of the approach to sample size calculation for obtaining empirical confidence interval of the required width with the predefined probability and give a framework of its general implementation. We illustrate the approach in Normal, Poisson, and Binomial distributions. The numeric results showed that the sample size necessary to obtain the required width of empirical confidence interval with the standard probability of $0.8$ or $0.9$ may be more than 20\% larger than the sample size calculated for the expected values of the parameters. 
\end{abstract}

\section{Notations}\label{sec:notations}
$A$ = scalar parameter of interest\\
$B$ = other parameters defining "theoretical" distribution corresponding current design\\
$a$ = assumed "true" value of A \\
$b$ = vector of the assumed values of B\\
$\alpha$ = required significance level\\
$CI= CI(\alpha) = (1-\alpha)$ confidence interval for $ a $\\
$h$ = a procedure for defining $d(CI)$\\
$C$ = full set of statistics for $h$\\
$ \Psi(C) $ - distribution of C ||
$d=d(CI)$ = width of $CI$ as defined by $h$ \\
$G(d)$ = distribution of $d(CI)$\\
$d_0$ = required width of $CI$\\
$\psi $ = probability of $d(CI)<d_0$ (analogue of power)\\
$\psi_0 $ = required probability of $d(CI)<d_0$ \\
$n_0$ =minimum sample size providing $ \psi\ge\psi_0 $ \\ 

\section{Introduction}\label{sec:intro}

	Each study is (or should be) carefully designed. The design should specify the goals of the study, measured parameters, method of measurement, sampling scheme, analysis, and at last, sample size  $N$. For estimation  studies, the goal is to obtain an estimation of the parameter of interest with a given precision. The precision is usually estimated by width of an empirical Confidence Interval ($CI$) that uses the estimation(s) of the parameter(s) from the study. The design should specify all details of the confidence interval: type (symmetric, the shortest, mid-point, etc.), one-sided or two-sided, and a confidence level $1- \alpha$ (usually $\alpha$ is $0.1$, $0.05$, or $0.01$). The width of an empirical $CI$ depends on the assumed distribution of the parameter of interest, and on the sample estimates of the parameters. These estimates are never known at the stage of design. However, a researcher assume some "true" values of the parameters using information from similar studies. This gives full definition of distribution function. 
	Many software tools, like PASS \cite{PASS}, MINITAB \cite{Minitab}, STATA \cite{STATA}, WinPepi \cite{WinPepi},  SURVEYSYSTEM \cite{SurveySystem},  QUATRICS \cite{QUATRICS}, R packages binomSamSize \cite{binomSamSize},  samplingbook \cite{samplingbook} and others use the assumed values instead of future estimates of the parameter(s) by substituting them in the corresponding formula for width of empiric $CI$as it is recommended in many popular text-books Yau \cite{Yau}, Machin \cite{Machin}, Lemeshow \cite{Lemeshow}. 

These procedures use only the confidence level, the desired width of corresponding confidence interval and expected values of the parameter(s) and their SDs (some of the procedures also permit finite population correction.)  The resulting sample size looks like it guarantees that the $CI$ that will be obtained in the study  will be definitely as small as required. Clearly, this is logically wrong because the empirical $CI$ is based not on the assumed values of the parameters but on their sample estimates. Thus the width of the empirical $CI$ is random. It differs from the assumed one and may be narrower or wider than the assumed value. Therefore the required precision may be not achieved.The correct approach is known for at least for 30 years (Greenland\cite{Greenland},Bristol\cite{Bristol},Beal\cite{Beal},Moore\cite{Moore},Grieve\cite{Grieve} ) but evidently much earlier. In 2003 it was revived by Jirotek et al.\cite{Jirotek} and especially by Kelley and Maxwell \cite{Kelley2003a} under the name AIPE (Accuracy In Parameter Estimation). Nevertheless, the program realization in commercial software are rare. The happy exception is SAS Proc POWER\cite{SAS}.  It defines the goal as "computing the probability of achieving the desired precision of a confidence interval, or the sample size required to ensure this probability" \cite[Proc POWER]{SAS}. In explaining the analysis of Confidence Interval in Proc POWER,  SAS states: "An analysis of confidence interval precision is analogous to a traditional power analysis, with $CI$ Half-Width taking the place of effect size and Prob(Width) taking the place of power". Unfortunately, in the current version 9.4, Proc POWER calculates the necessary sample size only for means, but not even for proportion and  other parameters \cite[STAT 14.1]{SAS}.In R the package MBESS realized AIPE\cite. these ideas were revivedAlso R package 
	In this text we describe the general approach to sample size evaluation for obtaining the width of the empirical $CI$ with a predefined probability, describe its implementation in R, and give examples for Normal, Poisson and Binomial distributions. 
	
\section{Approach}

We are looking for an estimate of a true value $a$ of a scalar parameter $A$ of a distribution $F(A,B)$, where $B$ is a set of additional parameters. The precision of the estimate is measured by the width $d$ of the confidence interval $CI=CI(\alpha)$ with confidence level $(1-\alpha) $. The $CI(\alpha)$ should be estimated using a sample $S$ of $N(S)$ observations from distribution $F(A,B)$. Let $h$ be a procedure for calculating the width $ d $ of $CI(\alpha)$ using the sample $S$. Usually there are several procedures to select from. For example, we can choose between several formula for $CI$ for a parameter $p$ of Binomial distribution $Bin(N,p)$: Clopper--Pearson, Normal, Wilson etc. To calculate the width of $CI$, procedure $h$ uses some statistics, like sample mean, sample standard deviation etc. Let $C\in \R^k$ be a set of these statistics.  Given $h(C)$ and fixing $N=N(S)$, we obtain the width $d=d(N,C)$ as a statistic with some distribution $G(d)$ .

For a given precision $d_0$ and probability $\psi_0$, we aim to find a sample size $n_0$ such that $d(N,C)\le d_0$ for any $N>n_0$ with probability $\psi_0$ or higher. Parameter $\psi_0$ plays a role similar to power in hypotheses testing.

  Let $N$ be given.  
\begin{Def}
	For a subset $R\subset\R^k$ define
	\begin{equation}\label{eq:def psi(R)}
	\psi_N(R)=Prob\left[C(S) \in R|\#S=N\right].
	\end{equation}
	
	Define $d_{min}(N)$ as 
	\begin{equation}\label{def:dmin}
	d_{min}(N)=\min\left\{d| \psi_N\left(\{d(N,C)\le d\}\right)\ge \psi_0 \right\}.
	\end{equation}
\end{Def}  

Evidently, it is enough to find $ n_0 $ such that $ d_{min}(N)<d_0 $ for all $ N>n_0 $.
\begin{Prop}
Let $ n_0 $ such that $ d_{min}(N)<d_0 $ for all $ N>n_0 $. Then $d(N,C)\le d_0$ for any $N>n_0$ with probability $\psi_0$ or higher.
\end{Prop}
\begin{proof}
	Indeed, let $N>n_0$. Then $d_{min}(N)<d_0$, and, therefore the probability that $d(N,C)\le d_0$ is at least the probability that $d(N,C)\le d_{min}(N)$. By definition of $d_{min}(N)$, the latter  is at least $\psi_0$.
\end{proof}

\subsubsection{The set $R_0(N)$}\label{sssec:def of R0}
The main difficulty consists of finding $d_{min}(N)$ in \eqref{def:dmin}. 
We solve it by considering explicitly the corresponding set $R_0(N)$.
\begin{Def}\label{def:def R0} 
	We define $R_0(N)$ as 
	\begin{equation}\label{def:def2 R_0}
	R_0(N)=\{d(N,C)\le d_{min}(N)\}.
	\end{equation}
\end{Def}

In other words, first, $R_0(N)$ must be "big" in a sense that the probability to obtain a result in $R_0(N)$ should be high (above $ \psi_0 $), and, simultaneously, "greedy" in a sense that 
for any point outside of $R_0(N)$, the width of an empiric $CI(C,N)$ is larger than the maximum of the width for any point in $R_0(N)$.

\begin{Rem}
In Definition~\ref{def:def R0} the statistic $C$ can be multidimensional, $C\in \R^k$, and, correspondingly, the set $R_0(N)$ will be a subset of $\R^k$.

If the statistic $C$ has the form $C=(c,B)$, where $c$ is a scalar and $B$ is the set of additional statistics, then one can consider a majorant \begin{equation}\label{eq:d*min}
d^*_{min}(N)=\min\left\{d| \psi_N\left(\{C|\max_B \left(d(N,(c,B))\right)\le d\}\right)\ge \psi_0 \right\}.
\end{equation}

In other words, consider the function $d^*(N,c)=\max_B {d(N,(c,B))}$. Then, considering $d^*(N,c)$ as a function on $\R^k$, we get 
$$
\{C|d(N,(c,B))\le d\}\subset\{C|d^*(N,c)\le d\},
$$
so $d_{min}(N)\le d^*_{min}(N)$. 

Alternatively, one can define $d^*_{min}(N)$ by using one-dimensional settings, with $d^*(N,c)$ (considered as univariate function) and  the marginal distribution $\tilde{\psi}_N(c)$  of $\psi_N(c,B)$.
\end{Rem}


In most practically important cases, the statistic $C$ is a scalar and the function $d(N,C)$ has a simple structure, e.g. monotone, bell-shaped etc., and the set $R_0(N)$  is an interval (for monotone $d(N,C)$) or two intervals (for bell-shaped case).

The described approach eliminates contradictions in the logic of calculating the sample size between hypotheses testing and estimation studies. The probability $\psi_0$ plays the role analogous to the power in hypothesis testing studies. 

Here, we address two immediate questions. First, we show how to implement it. Second, we compare between the sample sizes calculated by the proposed approach and the ones calculated by the current approaches.

\section{Implementation.} 
The implementation is simple. The common structure of a procedure for any distribution and any sampling is straightforward.
 
	1. Find initial estimate of $N$ . For many cases the good choice for starting value of $N$ is the sample size corresponding to the required width and the hypothesized value(s) of the parameter(s).
	 
	2. Find subset $R_0(N)$  for the estimated sample size N.
	
	3. Find $d(N)=\max_{C\in R_0(N)}(d(C,N))$ 
	
	4. If  $ d(N) > d_0 $ increase $N$ and repeat steps 2 and 3.
	
	5. Put $n_0$ equal to the last value of $ N $

\section{Examples.}
\subsection{Normal distribution.}
	Standard expression of CI for expectation $\mu $ of the Normal distribution $\textit{N}(\mu,\sigma)$ using I.i.d. sample of size $N$ leads to the following expression for width of the empiric confidence interval 	
	 \begin{equation}label{eq:def widthnorm}
	 d=2s \frac{t(1-\alpha/2,(N-1))}{\sqrt{N}} 
	\end{equation} 
								
	where $ t(1-\alpha/2, (N-1))  $ is the $ 1-\alpha/2 $ quantile of Student $ t $ distribution with $ (N-1)$  degrees of freedom, $ s $ is an estimated standard deviation, sample variance $ s^2 $ is defined as $ Q/N $, where $ Q=\Sigma(x_i- m)^2 $ , and $ m $ is the sample mean. The sum of squares  $ Q $ is distributed as $ \sigma^2\chi^2(N-1) $ where $ \sigma^2 $ is the assumed variance and $ \chi^2(N-1) $ is a chi-sqaure distribution with $ (N-1) $ degrees of freedom. The variance of the sample mean $ m $  is $ s^2/N $. 
	
	Thus  $ d $ is a monotone increasing function of  $ s $. 
	To be sure with probability $ \psi $ that the empirical interval will be shorter than $ d_0 $ for a fixed sample size $ N $ we should use $ s_0 = \sigma\sqrt{\chi^2(\psi,N) } $ , where $ \chi^2(\psi,N) $ is the $ \psi $ quantile of the  $\chi^2(N-1)) $ distribution. 
	  The given width $ d_0 $ will be guaranteed with the probability $ \psi $ if 
	\begin{equation}label{eq:def dnorm1}
	Prob\Big( 2s \frac{t(1-\alpha/2,N-1)}{\sqrt{N}} < d_0 \Big) \geq \psi,
	\end{equation}
	i.e.
	\begin{equation}label{eq:def dnorm2}
	Prob\Big(s < \frac{d_0\sqrt{N}}{2t(1-\alpha/2, N-1)} \Big)  \geq \psi
	\end{equation}
	and 
	\begin{equation}label{eq:def dnorm3}
	Prob\Big(s^2 < \frac{Nd_0^2}{4t^2(1-\alpha/2, N-1)} \Big)  \geq \psi
	\end{equation}
	but  $s^2 $ is distributed as $ \sigma^2\chi^2(N-1)/N $
	For any fixed $ \psi $, $ \chi^2(\psi,N)/N^2 $ is a monotone decreasing function of $ N $, and $ qt(1-\alpha/2, (N-1))  $ is a monotone decreasing function of N for any $\alpha$.
	Therefore the necessary sample size $ n_0 $ is the smallest solution of the inequality
	\begin{equation}label{eq:def dnorm4}
	\frac{t^2(1-\alpha/2, N-1)\chi^2(\psi,N-1)}{N^2} < \frac{d_0^2}{4\sigma^2}	         			
	\end{equation}
	
In Table 1 we present sample size $ (n_0) $  with corresponding coverage probability (Cov) and percent (Pow) of $CI$ with width less than the required (Width0) for postulated value of variance (named "Expected") and the variance found by the described algorithm (named "Exact"). The should be compared with the required coverage (named (1-alpha)) and the required proportion (named "Power")

\begin{landscape}
\begin{center}
Table 1.Normal distribution\vskip1cm
\begin{tabular}{ |c|c|c|c|c|c|c|c|c|c } 
		\hline
	\phantom{\Bigg(}Width0\phantom{\Bigg)} &1-alpha&Power&$n_0$Exp&CovExp&PowExp&$n_0$Exa&CovExa&PowExa\\
	\hline
	0.50000 &0.95 &0.8 &62 &0.9478 &0.4552 &73 &0.9489 &0.8135  \\
	0.50000 &0.95 &0.9 &62 &0.9491 &0.4422 &78 &0.9456 &0.9177  \\
	0.50000 &0.90 &0.8 &44 &0.9037 &0.4807 &53 &0.9001 &0.8299  \\
	0.50000 &0.90 &0.9 &44 &0.8921 &0.4791 &57 &0.9016 &0.9170  \\
	0.25000 &0.95 &0.8 &246 &0.9518 &0.4745 &267 &0.9453 &0.8144  \\
	0.25000 &0.95 &0.9 &246 &0.9526 &0.4682 &276 &0.9507 &0.9026  \\
	0.25000 &0.90 &0.8 &174 &0.8961 &0.4883 &190 &0.9036 &0.7972  \\
	0.25000 &0.90 &0.9 &174 &0.9007 &0.4871 &198 &0.9062 &0.8998  \\
	0.12500 &0.95 &0.8 &984 &0.9504 &0.4839 &1023 &0.9527 &0.8077 \\
	0.12500 &0.95 &0.9 &984 &0.9494 &0.4983 &1042 &0.9491 &0.9046 \\
	0.12500 &0.90 &0.8 &693 &0.8985 &0.4848 &725 &0.9041 &0.8022 \\
	0.12500 &0.90 &0.9 &693 &0.9007 &0.4886 &742 &0.8965 &0.9075 \\
	0.06250 &0.95 &0.8 &3934 &0.9501 &0.4899 &4010 &0.9479 &0.7986  \\
	0.06250 &0.95 &0.9 &3934 &0.9484 &0.4942 &4049 &0.9488 &0.9025  \\
	0.06250 &0.90 &0.8 &2771 &0.9055 &0.4979 &2835 &0.8974 &0.8035  \\
	0.06250 &0.90 &0.9 &2771 &0.9004 &0.4949 &2867 &0.9004 &0.9043  \\
		 \hline
\end{tabular}
\end{center}
\end{landscape}
\subsubsection{Comments}
	
	It can be seen that the “expected” sample size substantially underestimates the necessary sample size. Relative difference varies from 0.5 for power=0.9 and large width to 0.015 for power=0.8 and small width. This is due to the decrease of the ratio of standard deviation of chi-square distribution to its expectation when d.f. grows. Thus the sample size calculated for the expected value of the variance for D=1 and 1-beta=0.9 should be increased by 50\% from 16 to 24 to provide the declared probability of 0.9 to obtain the required width of 95\% CI. 
									
	\subsection{Poisson distribution.}
	Poisson distribution is a one parametric distribution on the set of all non-negative integers with probability function
	\begin{equation}label{eq:def Poiss1}
	P(n|\lambda)=\frac{\lambda^n}{n!}e^{-\lambda}	 	         			
	\end{equation}
									
	Poisson random variable mostly appears as a result of observing number of events in Poisson process. The intensity $ e $ of a process is assumed to be known and the question of “sample size” takes a form of “how long we have to observe the process to get a good interval estimation of the intensity?”. For example N is the total number of person-years of follow up and $ e $ is the incidence of a disease. The expected number of outcomes will be  $ \lambda = e*N $. Under some well-known assumptions, the observed number of events will have Poisson distribution with parameter $ \lambda $. Our goal is to find a sample size $ n_0 $, such that $ 1-\alpha $ confidence interval for the parameter $ e $ will have width $ d $ less or equal to $ d_0 $ with probability $ \psi_0 $. In a standard approach to calculating the necessary sample size, we assume some fixed value  $ e_0 $ of the parameter of interest. After obtaining a sample with size $ N $ we consider the observed number of outcomes x as an estimate of expected number of outcomes $ \lambda $ and calculate the sampling estimation of $e $ as $ x/N $. Correspondingly the width  $  d(e)=d(x)/N $. Thus we are looking for a solution of  $ n_0=min(N) $ such that 
	\begin{equation}label{eq:def Poiss2}
	P\Big(\frac{d(x,N)}{N} < d_0(e)\Big)\geq\psi 				
	\end{equation}
	for any $ N>n_0 $ .
	There are many expressions for $CI$ of the parameter x of Poisson distribution (Patil and Kulkarni \cite{Patil} considers 19).
	One of expressions, recommended by \cite{Patil} for $CI$ of the parameter $ \lambda $ of Poisson distribution is 
	 Garwood (1994)\cite{Garwood}, 
	 \begin{equation}label{eq:def Poiss3}
	 (\chi^{2}(2x,\alpha_1),\chi^{2}(2x+2,\alpha_2))
	 \end{equation}
	Usually,  $ \alpha_2=1-\alpha_1, \alpha_1=\alpha/2. $ 
	Then 
	\begin{equation}label{eq:def Poiss4}
	d(x,\alpha)=(\chi^{2}(2x+2,1-\alpha/2) - \chi^{2}(2x,\alpha/2)) 
	\end{equation}
	and correspondingly
	\begin{equation}label{eq:def Poiss5}
	d(e,\alpha)=(\chi^{2}(2x+2,1-\alpha/2) - \chi^{2}(2x,\alpha/2))/N 
	\end{equation}

	For any fixed  $ x $, $ d(e,N) $ is a monotone decreasing function of $ N $ and for any fixed  $ N $, $ d(e,N) $ is a monotone increasing function of $ x $. Therefore we are looking for a set $ R_0(\lambda,\psi) $ such that $ Prob(\lambda)(R_0)\geqq \psi $ with the lowest $ max{(x \in R_0)}$.Thus the set $ R_0(\lambda,\psi) $ is an interval $ {0,1,2,..Q(\psi,\lambda) } $
	From the definition of Poisson(m) we can see that for any $ \lambda >=1 $, the probability to observe number of events in Poisson(m) not greater than  $\lambda $ is below 0.8 (for $\lambda =1 P(x<=1)= .73575888 $ and this probability monotone decreasing with $\lambda $).
	Thus the $ Q(\psi,\lambda )>\lambda $ for $ \psi > 0.75 $ and $ m\geqq1 $ , $  d(Q(\psi,\lambda),N)>d(\lambda,N) $ leading to "expected"  sample size will be smaller than the "empiric" one.

	\subsubsection{Results}
	Table 2 has the same structure as Table 1. It presents sample size ($n_0$) with corresponding coverage probability (Cov) and percent (Pow) of $CI$ with width less than the required (Width0) for postulated value of variance (named "Expected") and the variance found by the described algorithm (named "Exact"). The should be compared with the required coverage (0.95) and the required proportion (named "Power")

\begin{landscape}
    \begin{center}
    Table 2. Poisson distribution.\vskip1cm
 	\begin{tabular}{ |c|c|c|c|c|c|c|c|c| } 
		\hline		
	\phantom{\Bigg(}Rate\phantom{\Bigg)} & Width0 & Power & $n_0$Exp & CovExp & PowExp & $n_0$Exa & CovExa & PowExa \\
	\hline
	0.01 & 0.002 & 0.8 & 39439 & 0.9532 & 0.5055 & 41064 & 0.9589 & 0.8070 \\
	0.01 & 0.002 & 0.9 & 39439 & 0.9534 & 0.5092 & 41861 & 0.9515 & 0.9045  \\
	0.01 & 0.001 & 0.8 & 155683 & 0.9479 & 0.4929 & 158936 & 0.9491 & 0.8072 \\
	0.01 & 0.001 & 0.9 & 155683 & 0.9529 & 0.4917 & 160630 & 0.9537 & 0.9050  \\
	0.02 & 0.004 & 0.8 & 19719 & 0.9538 & 0.5064 & 20532 & 0.9561 & 0.8162  \\
	0.02 & 0.004 & 0.9 & 19719 & 0.9548 & 0.4972 & 20931 & 0.9512 & 0.9042  \\
	0.02 & 0.002 & 0.8 & 77841 & 0.9562 & 0.5008 & 79468 & 0.9500 & 0.8069  \\
	0.02 & 0.002 & 0.9 & 77841 & 0.9499 & 0.5040 & 80315 & 0.9540 & 0.9047  \\
	0.04 & 0.008 & 0.8 & 9859 & 0.9521 & 0.5120 & 10266 & 0.9513 & 0.8112 \\
	0.04 & 0.008 & 0.9 & 9859 & 0.9535 & 0.5066 & 10466 & 0.9530 & 0.9014  \\
	0.04 & 0.004 & 0.8 & 38920 & 0.9498 & 0.5011 & 39734 & 0.9490 & 0.7996  \\
	0.04 & 0.004 & 0.9 & 38920 & 0.9502 & 0.4944 & 40158 & 0.9509 & 0.8971  \\
	0.08 & 0.016 & 0.8 & 4929 & 0.9534 & 0.5072 & 5133 & 0.9559 & 0.8047  \\
	0.08 & 0.016 & 0.9 & 4929 & 0.9532 & 0.5034 & 5233 & 0.9490 & 0.9037  \\
	0.08 & 0.008 & 0.8 & 19460 & 0.9532 & 0.5021 & 19867 & 0.9510 & 0.8057  \\
	0.08 & 0.008 & 0.9 & 19460 & 0.9510 & 0.4965 & 20079 & 0.9492 & 0.9075  \\
	 \hline
\end{tabular}

\end{center}
\end{landscape}
\subsubsection{Comments}
	It can be seen that the Garwood formula provides perfect coverage. However, the "expected" sample size 	is lower that the "exact" by 2\% - 5\% and provides the "expected power" around  50\%  while the "exact" sample size guarantees the requested values of power.
	If one prefers to avoid using statistical software, it is possible to use simple algebraic approximations for $ Q(\psi,\lambda) $ and $ CI $.
	The value $ Q(\psi,\lambda) $ may be approximately estimated by using direct and inverse transformation of Poisson variable to Normal using Anscombe \cite{Anscombe} approximation 
    \begin{equation}
	z=2\sqrt{(x+0.375)}~{}N(\sqrt(x0),1) 
	\end{equation}
	and recommended inverse transformation  
	\begin{equation}
	 x=\frac{z^2}{4}-0.125  
	\end{equation}
	From these we obtain
	\begin{equation}
	Q(\psi,\lambda)\approx(2\sqrt{(Ne_0 +0.375)}+Z(\psi )^2-0.125
	\end{equation}

	\subsection{Binomial distribution.} 
	Here we are interested in finding the sample size $ n_0 $ that provides the desired width $ d_0 $ for assumed proportion $ p_0 $ of the number $ x $ of outcomes 1 in a sample of size $ N $.
	There are many formula for $CI$ of the parameter $ p $ of the Binomial distribution. One of the most usable\cite{Newcombe} is the Wilson’s interval\cite{Wilson}

	\begin{equation}label{eq: Wilson1}
		d(CI)=\frac{ 2N\hat{p}+Z_{\alpha}^2 \pm Z_{\alpha}^2\sqrt{4N\hat{p}(1-\hat{p}) +Z_{\alpha}^2 }}{2N+ 2Z_{\alpha}^2} 
	\end{equation}
	
	where $ \hat{p} $ is the observed proportion $ o=x/N $ and x is the number of outcomes 1. Its width is 
	\begin{equation}label{eq: Wilson2}
	D(\hat{p},N)=\frac{ 2Z_{\alpha}^2\sqrt{4N\hat{p}(1-\hat{p}) +Z_{\alpha}^2 }}{2N+ 2Z_{\alpha}^2}  
	\end{equation}
	The maximum width corresponds to maximum $ \hat{p}*(1-\hat{p}) $. Thus for any $ N $, if $ abs(\hat{p}_1-0.5)>abs(\hat{p}_2-0.5) $ then $ d(\hat{p}_1,N)< d(\hat{p}_2,N) $. Our goal is to guarantee that $ Prob(D<d_0)\geqq \psi $ with the smallest $ N $. Following the “worst – best” scenario, we are going to find a set $ R_0 $, such that      
	\begin{equation}label{eq: Wilson3}
		Prob(R_0|p,N)\geqq \psi					
	\end{equation}	
	and 
	\begin{equation}label{eq: Wilson4}
	R_0=\argmin_R \max_{\hat{p} \in R}(abs(\hat{p}-0.5))). 
	\end{equation}
	
	Evidently, this set has a form  ${0,q}\bigcup{(1-q),1} $ where $q$ is defined from equation \ref{Wilson4}. However, for $p$ not too close to $0.5$ and reasonable sample size $N$ one of intervals has negligible probability under $Bin(p,N)$. For examples, the upper interval $\left((1-q),1\right)$ will have the probability below $0.001$ for $p_0 =0.2$ if $N>20$ , for $p_0=0.35$ if $N>40$, for $p_0=0.4$ if  $N>70$ and for $p_0= 0.45$ if $N>260$. Our simulations demonstrated that for $p<0.45$ and  $R_0$ of the form $\{0,q\}$ the probability  $P(N,q)$ will be greater than the probability $P({(1-q),1})$  for $N$ equal to necessary sample size. Therefore we will use $R_0$ of the form of one interval.
\subsubsection{Results}
Table 3 has the same structure as table 1. It presents the sample size ($n_0$) with corresponding coverage probability (Cov) and percent (Pow) of $CI$ with width less than the required (Width0) for postulated value of the parameter $p_0$ (named "Expected") and the variance found by the described algorithm (named "Exact"). The should be compared with the required coverage (0.95) and the required proportion (named "Power")
\begin{landscape}
\begin{center}
	Table 3. Binomial distribution.\vskip1cm
	
	\begin{tabular}{ |c|c|c|c|c|c|c|c|c| } 
		
		\hline
\phantom{\Bigg(}$p_0$\phantom{\Bigg)}&Width0&Power&$n_0$Exp&PowExp&CovExp&S$n_0$Exa&PowExa&CovExa\\
 \hline
0.5000 & 0.10& 0.8 & 381 & 1.0000 & 0.944 & 381 & 1.0000 & 0.947\\
0.5000 & 0.10& 0.9 & 381 & 1.0000 & 0.943 & 381 & 1.0000 & 0.948\\
0.5000 & 0.05& 0.8 & 1533 & 1.0000 & 0.951 & 1533 & 1.0000 & 0.950\\
0.5000 & 0.05& 0.9 & 1533 & 1.0000 & 0.947 & 1533 & 1.0000 & 0.949\\
0.2500 & 0.10& 0.8 & 286 & 0.5078 & 0.945 & 302 & 0.8253 & 0.955\\
0.2500 & 0.10& 0.9 & 286 & 0.5031 & 0.944 & 309 & 0.9061 & 0.957\\
0.2500 & 0.05& 0.8 & 1150 & 0.5005 & 0.951 & 1182 & 0.8134 & 0.951\\
0.2500 & 0.05& 0.9 & 1150 & 0.4961 & 0.949 & 1199 & 0.8989 & 0.952\\
0.1250 & 0.10& 0.8 & 170 & 0.5351 & 0.953 & 192 & 0.8397 & 0.938\\
0.1250 & 0.10& 0.9 & 170 & 0.5378 & 0.950 & 201 & 0.9097 & 0.957\\
0.1250 & 0.05& 0.8 & 674 & 0.5091 & 0.952 & 722 & 0.8227 & 0.952\\
0.1250 & 0.05& 0.9 & 674 & 0.5203 & 0.953 & 745 & 0.9145 & 0.949\\
0.0625 & 0.10& 0.8 & 98 & 0.5850 & 0.945 & 121 & 0.8624 & 0.944\\
0.0625 & 0.10& 0.9 & 98 & 0.5839 & 0.948 & 132 & 0.9326 & 0.953\\
0.0625 & 0.05& 0.8 & 369 & 0.5564 & 0.961 & 424 & 0.8479 & 0.958\\
0.0625 & 0.05& 0.9 & 369 & 0.5484 & 0.960 & 449 & 0.9293 & 0.952\\
 \hline
\end{tabular}
\end{center}
\end{landscape}
\subsubsection{Comments}
The first fact that we can see from table 3 is that for expected probability  $p_0=0.5$ the results strongly differ from all others. The thing it that  in our best-worst scenario the set $R_0$ always contains the value $0.5$.Therefore the sample size does not depend on power $0.8$ or $0.9$ and coincides with expected power for $p_0=0.5$. The power, i.e the proportion of $CI$ with width less than the required, is $1.0$ because the width of $CI$ for $p=0.5$ is bigger than for any other $p$. The coverage is close to nominal value of $0.95$ supporting good properties of Wilson $CI$. However, the situation is different for other values of $p_0$. The "exact" sample size is bigger than the "expected" one. The relative difference may be more than 20\% (for $p_0=0.0625$, $d_0=0.05$ and power$=0.9$). The "exact' sample size provides power not less than the required, while the power for "expected" sample size is around $0.5$.

\newpage	
\lstset{basicstyle=\small,tabsize=2}
\section{Appendix. R scripts}
\subsection{Normal distrbution}
\begin{lstlisting}
# starting value using expected variance 
# and Normal approximation
nstart <- function(a,sigma,d) {
	za <-qnorm(1-a/2)
	nst <- (2*za*sigma/d)^2
	return(ceiling(nst))
}

# estimation of sd(power) using chi-square 
# quantile and expected variance
parpow <- function ( power, n, sigma) {
	sd <- sigma*sqrt(qchisq(power, n-1))
	return(sd)
}
# width of (1-a) CI using t distribution 
width <- function( a,n,sd) {
	ta <- qt(1-a/2,n-1)
	width <- 2*ta*sd/sqrt(n)
	return(width)
}
# test power and coverage ***********************
testnorm <- function(nsim, ssize,width, sigma, alpha) {
	ns <-rep(ssize,nsim)
	myse <- sigma*sqrt(rchisq(nsim,ssize))/ssize
	mymean <- rnorm(nsim,0)*sigma/sqrt(ssize)
	ta <- qt(1-alpha/2,ssize-1)
	lo <- mymean-myse*ta
	up <- mymean+myse*ta
	cover <- mean(lo*up<0)
	small <- mean((up-lo)<width)
	res <- c(cover,small)
	return(res)
}


# sample size using chi-square and t distribution
ssnorm <- function(a, power,sigma, dmax) {
	n0  <- nstart(a, sigma, dmax)
	sd_0 <- parpow(power,n0,sigma)
	w0  <- width(a,n0,sd_0)
	step <- 2 * (w0 >dmax) -1
	w <-w0
	s0 <-n0
	while ( step*(w>dmax)>0 )  {
		qchisq(power, s0-1) -> qb 
		qt(1-a/2,s0-1) -> ta
		2*ta*sqrt(qb)/s0 -> w 
		s0+step -> s0 
	}
	ss<-ceiling(s0)
	return(ss)
}


ssnormapr <-function(a,prot,sd,dmax) {
	qnorm(1-a/2) -> za
	qnorm(prot)  -> zb
	n0  <- nstart(a, sd, dmax)
	dmax/sd -> dn 
	10*dn -> d_0 
	floor(4*(za*dn)^2) -> sexp 
	w0  <- width(a,n0,sd)
	sexp -> s0 
	step <- 2 * (w0 >dmax) -1
	while    (d_0 > dn) {
		((sqrt(2*s0)+zb)^2)/2 -> qb
		2*za * sqrt(qb)/s0 -> d_0 
		s0+step -> s0 
	}
	ss<-s0+2
}

tabnorm <- function(k,nsim)  {
	tab <-matrix(nrow=4*k,ncol=10)
	colnames(tab) <- c("Width","1-alpha", "Power",
	 					"Expected","CovExp","PowExp", "Exact", 
	 					"CovExact","PowExact", "Approximate")
	row<-0
	for (r in 1: k ) {
		width <-1/2^r
		for (a in c(0.05, 0.1)) {
			for (b in c(0.8, 0.9)) {
				row <-row+1
				tab[row,1] <-width 
				tab[row,2] <- 1-a
				tab[row,3] <- b
				tab[row,4] <- nstart(a, 1, width)
				tab[row,5] <- testnorm(nsim,tab[row,4],width,1,a)[1]
				tab[row,6] <- testnorm(nsim,tab[row,4],width,1,a)[2]
				tab[row,7] <- ssnorm(a, b, 1, width)
				tab[row,8] <- testnorm(nsim,tab[row,7],width,1,a)[1]
				tab[row,9] <- testnorm(nsim,tab[row,7],width,1,a)[2]
				tab[row,10] <- ssnormapr(a, b, 1, width)
			}
		}
	} 
	return(tab)
}

z<-tabnorm(5,10000)
z

\end{lstlisting}

\subsection{Poisson distribution}
\begin{lstlisting}
exactPoiCI <- function (x, alpha) {
	upper <- 0.5 * qchisq((1-(alpha/2)), (2*x+2))
	lower <- 0.5 * qchisq(alpha/2, (2*x ))
	return(c(lower,upper))
}

qpoisson <- function(p, lambda) {
	qn <- qpois(p, lambda, lower.tail = TRUE, log.p = FALSE)
}

ssexpectpoi <-function (rate, d_0, alpha) {
	startss <-ceiling(16*rate/d_0^2)
	jump <- ceiling(startss/10)
	step<-1
	dw0 <- 10*d_0
	s      <- startss
	lambda <-rate*s
	if (dw0>d_0) {
		while (dw0 > d_0)  {
			s <- s +jump
			lambda <- s*rate
			dw0 <- (exactPoiCI(lambda,alpha)[2] 
					- exactPoiCI(lambda,alpha)[1])/s
		}
	}
	while (dw0 <= d_0)  {
		s <- s -step
		lambda <- s*rate
		dw0 <-  (exactPoiCI(lambda,alpha)[2] 
					- exactPoiCI(lambda,alpha)[1])/s
	}
	return(s)
}

ssexactpoi <-function (rate, d_0, alpha, power) {
	startss <-ceiling(16*rate/d_0^2)
	jump <- ceiling(startss/10)
	step<-1
	dw0 <- 10*d_0
	s      <- startss
	while (dw0 > d_0)  {
		s <- s +jump
		lambda <- s*rate
		nbeta <- qpois(power, lambda)
		dw0 <-  (exactPoiCI(nbeta,alpha)[2] 
				- exactPoiCI(nbeta,alpha)[1])/s
	}
	s <- s +jump
	lambda <- s*rate
	nbeta <- qpois(power, lambda)
	dw0 <-  (exactPoiCI(nbeta,alpha)[2] 
			- exactPoiCI(nbeta,alpha)[1])/s
	while (dw0 <= d_0)  {
		s <- s -step
		lambda <- s*rate
		nbeta <- qpois(power, lambda)
		dw0 <-  (exactPoiCI(nbeta,alpha)[2] 
				- exactPoiCI(nbeta,alpha)[1])/s
	}
	s <- s +step
	lambda <- s*rate
	nbeta <- qpois(power, lambda)
	dw0 <-  exactPoiCI(nbeta,alpha)/s
	return(s)
} 

testpoi<- function(nsim, rate,d_0,nsam,alpha) {
	lambda <- rate*nsam
	x <-    rpois(nsim, lambda)
	n <- rep(nsam, nsim)
	ci0 <- exactPoiCI(x,alpha)/n
	res1 <- matrix(ci0,ncol=2)
	pcover <- mean(res1[,2]>rate & res1[,1]<rate)
	len <-  res1[,2]-res1[,1]
	plen <- mean(len<d_0)
	res <- c(plen,pcover)
	return(res)
}

tabpoi <- function(k,nsim)  {
	tab <-matrix(nrow=8*k,ncol=10)
	colnames(tab) <- c("Rate","Width","1-alpha", "Power",
						"SSExpected","CoverExp","PowExp", "Exact",
	 					"CoverExact","PowExact")
	row<-0
	for (r in 1: k ) {
		rate <- 0.005*2^r
		for (width in c(0.2*rate, 0.1*rate)) {
			for (a in c(0.05, 0.1)) {
				for (b in c(0.8, 0.9))  {
					row <-row+1
					tab[row,1] <- rate
					tab[row,2] <- width 
					tab[row,3] <- 1-a
					tab[row,4] <- b
					tab[row,5] <- ssexpectpoi(rate, width,a)
					tab[row,6] <- testpoi (nsim,rate,width, tab[row,5],a)[2]
					tab[row,7] <- testpoi (nsim,rate,width, tab[row,5],a)[1]
					nw<-ssexactpoi (rate, width, a, b) 
					tab[row,8] <- nw       
					tab[row,9] <- testpoi (nsim,rate,width, tab[row,8],a)[2]
					tab[row,10] <- testpoi (nsim,rate,width, tab[row,8],a)[1]
				}
			}
		}
	} 
	return(tab)
}

z<-tabpoi(4,10000)
z
\end{lstlisting}

\subsection{Binomial distribution}
\begin{lstlisting}
library(binom)
ssexpect <-function (p0, d_0, alpha, jump, step) { 
	dwils0 <- 10*d_0 
	za2    <- qnorm(1-alpha/2) 
	s = ceiling(za2^2*p0*(1-p0)/d_0/d_0) 
	minp0  <- min(p0,1-p0) 
	while (dwils0 > d_0)  { 
		s <- s +jump 
		dwils0 <- 2 * za2 * sqrt(minp0*(1-minp0)*s
					+ (za2^2)/4)/(s*(1+ za2^2/s)) 
	} 
	s <- s +jump 
	dwils0 <- d_0/2 
	while (dwils0 < d_0)  { 
		s <- s -step 
		dwils0 <- 2 * za2 * sqrt(minp0*(1-minp0)*s
					+ (za2^2)/4)/(s*(1+ za2^2/s)) 
	} 
	return(s+1) 
}  

sswilson <-function (p0, d_0, alpha, power, startss, jump, step) { 
	dwils0 <- 10*d_0
	za2    <- qnorm(1-alpha/2)
	zb     <- qnorm(power)
	s      <- startss
	minp0  <- min(p0,1-p0)
	while (dwils0 > d_0)  {
		s <- s +jump
		qbeta <- min(0.5,qbinom(power, s, minp0)/s)
		dwils0 <- 2 * za2 * sqrt(qbeta*(1-qbeta)*s
					+ (za2^2)/4)/(s*(1+ za2^2/s))
	}
	s <- s +jump
	dwils0 <- d_0/2
	while (dwils0 < d_0)  {
		s <- s -step
		qbeta <- min(0.5, qbinom(power, s, minp0)/s)
		dwils0 <- 2 * za2 * sqrt(qbeta*(1-qbeta)*s
					+ (za2^2)/4)/(s*(1+ za2^2/s))
	}
	return(s+1)
} 

testbin<- function(nsim, p0,d_0,nsam,alpha) { 
	x <-    rbinom(nsim, nsam, p0)
	n <- rep(nsam, nsim)
	ci0 <-  binom.confint(x, n, conf.level = 1-alpha,methods = "wilson")
	len <-   ci0$upper-ci0$lower
	plen <- mean(len<d_0)
	pcover <- mean(ci0$upper>p0 & ci0$lower<p0)
	res <- c(plen,pcover)
	return(res)
} 

tabbin <- function(k,nsim){ 
	tab <-matrix(nrow=8*k,ncol=10)
	colnames(tab) <- c("Probability","Width","1-alpha", "Power",
						"Expected","CoverExp","PowExp", "Exact",
						"CoverExact","PowExact")
	row<-0
	for (r in 1: k ){
		p0 <- 1/2^r
		for (width in c(0.1, 0.05)){
			for (a in c(0.05, 0.1)){
				for (b in c(0.8, 0.9)){
					row <-row+1
					tab[row,1] <- p0 
					tab[row,2] <- width 
					tab[row,3] <- 1-a
					tab[row,4] <- b
					tab[row,5] <- ssexpect(p0, width,a,10,1)
					tab[row,6] <- testbin (nsim,p0,width, tab[row,5],a)[1] 
					tab[row,7] <- testbin (nsim,p0,width, tab[row,5],a)[2] 
					nw<-sswilson (p0, width, a, b, 5, 10,1) 
					tab[row,8] <- nw       
					tab[row,9] <- testbin (nsim,p0,width, tab[row,8],a)[1] 
					tab[row,10] <- testbin (nsim,p0,width, tab[row,8],a)[2] 
				}
			}
	    }
	} 
 return(tab)
} 

z<-tabbin(2,10000)
z

\end{lstlisting}
\end{document}